\documentclass[12pt,reqno]{article}

\usepackage[usenames]{color}
\usepackage{amssymb}
\usepackage{amsmath}
\usepackage{amsthm}
\usepackage{amsfonts}
\usepackage{amscd}
\usepackage{graphicx}
\usepackage[shortlabels]{enumitem}

\usepackage[colorlinks=true,
linkcolor=webgreen,
filecolor=webbrown,
citecolor=webgreen]{hyperref}

\definecolor{webgreen}{rgb}{0,.5,0}
\definecolor{webbrown}{rgb}{.6,0,0}

\usepackage{color}
\usepackage{fullpage}
\usepackage{float}

\usepackage{graphics}
\usepackage{latexsym}
\usepackage{epsf}
\usepackage{breakurl}

\begin{document}

\theoremstyle{plain}
\newtheorem{theorem}{Theorem}
\newtheorem{corollary}[theorem]{Corollary}
\newtheorem{lemma}[theorem]{Lemma}
\newtheorem{proposition}[theorem]{Proposition}

\theoremstyle{definition}
\newtheorem{definition}[theorem]{Definition}
\newtheorem{example}[theorem]{Example}
\newtheorem{conjecture}[theorem]{Conjecture}

\theoremstyle{remark}
\newtheorem{remark}{Remark}

\title{The Simplest Binary Word with Only Three Squares}
\author{Daniel Gabric and Jeffrey Shallit \\
School of Computer Science \\
University of Waterloo \\
Waterloo, ON  N2L 3G1 \\
Canada\\
\href{mailto:dgabric@uwaterloo.ca}{\tt dgabric@uwaterloo.ca} \\
\href{mailto:shallit@uwaterloo.ca}
{\tt shallit@uwaterloo.ca}}

\date{}

\maketitle

\begin{abstract}
We re-examine previous constructions of infinite binary words containing few distinct squares with the goal of finding the ``simplest'', in a certain sense.  We exhibit several new constructions.  Rather than using tedious case-based arguments to prove that the constructions have the desired property, we rely instead on theorem-proving software for their correctness.
\end{abstract}

\section{Introduction}

One of the earliest results in combinatorics on words is that squares are unavoidable over a two-letter alphabet, but are avoidable over a three-letter alphabet \cite{Thue:1906,Thue:1912,Berstel:1995}.
Here a ``square'' is a nonempty word of the form $xx$, ``unavoidable'' means that every sufficiently long word contains a square subword, and ``avoidable'' means there exists an infinite word containing no squares.

Although squares are unavoidable over a two-letter alphabet,
Entringer, Jackson, and Schatz \cite{Entringer&Jackson&Schatz:1974} proved that there exist infinite binary words containing no squares of order $\geq 3$.   (The order of a square $xx$ is $|x|$, the length of $x$.)   This was later improved by
Fraenkel and Simpson; they showed the existence of binary words having only three distinct squares.

The main tool for creating such words is the {\it morphism}:   a map $h: \Sigma^* \rightarrow \Delta^*$ for alphabets
$\Sigma$, $\Delta$ obeying the rule $h(xy) = h(x) h(y)$ for all $x, y \in \Sigma^*$.  A morphism is $k$-{\it uniform\/} if $|h(a)| = k$ for all $ a\in \Sigma$.  If it is $k$-uniform for some $k$, then we say it is {\it uniform}.   A $1$-uniform morphism is called a {\it coding}.  If $\Delta \subseteq \Sigma$ we can iterate $h$, writing
$h^2(x)$ for $h(h(x))$, and so forth.   If further
$h(a) = ax$ for some $a \in \Sigma$, $x \in \Sigma^*$, and $h^i(x) \not= \epsilon$ for all $i$, then iterating $h$ infinitely produces an infinite word $h^\omega(a) = ax h(x) h^2(x) \cdots$ called a {\it fixed point\/} of $h$.
If an infinite word is the image, under a coding, of a fixed point of a $k$-uniform morphism, it is called {\it $k$-automatic}.   The {\it weight\/} of a morphism $h: \Sigma^* \rightarrow \Sigma^*$
is defined to be $\sum_{a \in \Sigma} |h(a)|$,
and the weight of a $k$-automatic infinite word is defined to be the weight of its defining morphism.

In this note we find the ``simplest'' infinite binary word having at most three distinct squares.   Our criterion for simplicity is as follows:  
\begin{enumerate}[(a)]
\item the word should be generated by a finite automaton of $s$ states taking the base-$k$ representation of $n$ as input (i.e., a $k$-automaton), most significant digit first; and
\item the product $k\cdot s$ should be as small as possible.
\end{enumerate}
By Cobham's theorem \cite{Cobham:1972}, this is same as saying the word is generated as the image, under a coding, of a fixed point of a $k$-uniform morphism over an alphabet of $s$ letters.  

One practical advantage to restricting our attention to $k$-automatic words is that the property of having exactly three distinct square factors can be stated in first-order logic, thus reducing the verification to a completely routine calculation using a decision procedure \cite{Charlier&Rampersad&Shallit:2012}.

\section{The Entringer-Jackson-Schatz contruction}
\label{two}

We begin with a description of the construction of Entringer-Jackson-Schatz.
  Here very slightly modified from the original, it starts with an arbitrary squarefree word $\bf z$ over $\{ 0,1,2 \}$ and applies the uniform morphism
\begin{align*}
    h(0) &= 1100 \\
    h(1) &= 0111 \\
    h(2) &= 1010
\end{align*}
to it.   They proved that the resulting word $h({\bf z})$ has no squares of order $\geq 3$; in fact, the only squares that appear are $0^2, 1^2, (01)^2, (10)^2$, and $(11)^2$.

Although this is indeed a simple construction, in terms of automatic sequences, it can be improved.   The minimum automaton size for $h({\bf z})$, over all $2$-automatic squarefree words $\bf z$, is $10$, as can be verified by breadth-first search, with pruning if the prefix constructed so far requires $11$ or more states.   

This minimum number of states is achieved, for example, by applying $h$ to the famous squarefree word 
${\bf vtm} := \tau(g^\omega(0)) = 2102012101202102012021012102012 \cdots$, where
\begin{align*}
g(0) & = 01 & \quad \tau(0) &= 2\\
g(1) & = 20 & \quad \tau(1) &= 1\\
g(2) & = 23 & \quad \tau(2) &= 0\\
g(3) & = 02 & \quad \tau(3) &= 1 .
\end{align*}
\begin{remark}  The word $\bf vtm$ is (up to renaming) the classical squarefree word of
Thue \cite{Thue:1912}.   It can be defined in many different
ways \cite{Berstel:1978}, including as the
fixed point of the morphism defined by
$2 \rightarrow 210$, $1 \rightarrow 20$,
$0 \rightarrow 1$.   The name {\bf vtm} for
this word comes from \cite{Blanchet-Sadri&Currie&Rampersad&Fox:2013}.
\end{remark}

A novel alternative construction (not necessarily an image of  $\bf vtm$) needs only six states.  This is the minimum possible number of states for a $2$-automatic word containing no squares of order $\geq 3$ and only 5 distinct squares.  
\begin{theorem}
Consider the infinite word
$\rho(f^\omega(0))$, where
\begin{align*}
f(0) &= 01 & \quad \rho(0) &= 0 \\
f(1) &= 23 & \quad \rho(1) &= 0 \\
f(2) &= 45 & \quad \rho(2) &= 0 \\
f(3) &= 02 & \quad \rho(3) &= 0 \\
f(4) &= 05 & \quad \rho(4) &= 1 \\
f(5) &= 25 & \quad \rho(5) &= 1 .
\end{align*}
This is the lexicographically least word generated by a $2$-automaton of $\leq 6$ states, containing no squares of order $\geq 3$, and
only 5 distinct squares.
\end{theorem}

\section{Only three distinct squares}

The Entringer-Jackson-Schatz construction was optimally improved by Fraenkel and Simpson \cite{Fraenkel&Simpson:1994}, as follows:  they constructed an infinite binary word containing only $3$ squares:  $0^2$, $1^2$, and $(10)^2$.

Their construction is rather complicated, and also has a complicated proof. It starts with an infinite squarefree word $\bf w$ over $\{0,1,2\}$ avoiding the subwords $020$ and $121$.   (Although they do not say so, an example of such a word is given by renaming the letters in ${\bf vtm} := \tau(g^\omega(0))$ above.)
Then replace every occurrence of $12$ with $132$.  Next, replace every remaining occurrence of $21$ with $241$.   Finally, apply the morphism 
$\alpha$ defined as follows:
\begin{align*}
\alpha(0) &= 011000111001 \\
\alpha(1) &= 011100011001 \\
\alpha(2) &= 011001110001 \\
\alpha(3) &= 01100010111001 \\
\alpha(4) &= 01110010110001 .
\end{align*}
The resulting word avoids all squares except $0^2$, $1^2$, and $(01)^2$.

Because of the inherent complexity of this construction, it seems desirable to find simpler ones.   An example using $24$-uniform morphisms was given by Rampersad et al.~\cite{Rampersad&Shallit&Wang:2005}.
Define 
\begin{align*}
p(0) &= 012321012340121012321234 \\
p(1) &= 012101234323401234321234 \\
p(2) &= 012101232123401232101234 \\
p(3) &= 012321234323401232101234 \\
p(4) &= 012321234012101234321234
\end{align*}
and 
\begin{align*}
\beta(0) &= 011100 \\
\beta(1) &= 101100 \\
\beta(2) &= 111000 \\
\beta(3) &= 110010 \\
\beta(4) &= 110001.
\end{align*}
Then $\beta(p^\omega(0))$ is an infinite
word containing only the squares $0^2$,
$1^2$, and $(01)^2$.  This construction gives a $24$-automatic sequence generated by an automaton of $18$ states, so its weight is $24\cdot 18 = 432$.

\subsection{Ochem's word}

Ochem \cite{Ochem:2006} provided a different construction in 2006:
\begin{align*}
\sigma(0) & = 00011001011000111001011001110001011100101100010111\\
\sigma(1) & = 00011001011000101110010110011100010110001110010111 \\
\sigma(2) &= 00011001011000101110010110001110010111000101100111
\end{align*}
He showed that if $\bf x$ is a $(7/4+\epsilon)$-free word, then
$\sigma({\bf x})$ contains only three squares.

In fact, we can also successfully apply $\sigma$ to the word $\bf vtm$ above, even though it is not $(7/4+\epsilon)$-free.    Since $\sigma$ is a uniform map, we know that $\sigma({\bf vtm})$ is
$2$-automatic.  

\begin{theorem}
This word $\sigma({\bf vtm})$ is a $2$-automatic word containing only three distinct squares.  It is generated by an automaton with 109 states (and has weight $2 \cdot 109 = 218$).
\end{theorem}

\subsection{The Harju-Nowotka construction}

Harju and Nowotka \cite{Harju&Nowotka:2006} generated an infinite binary word with three squares by defining the map
\begin{align*}
\zeta(0) &= 111000110010110001110010 \\     \zeta(1) &= 111000101100011100101100010  \\ \zeta(2) &= 111000110010110001011100101100 \ .
\end{align*}
and then applying it to $\bf vtm$.  

The morphism $\zeta$ is clearly not uniform.  However, the lengths of the images of $0, 1, 2$ are
(respectively) $24, 27, 30$ and form an arithmetic progression.  This is enough to show that $\zeta({\bf vtm})$ is $2$-automatic, as the following result shows.

\begin{theorem}
Let ${\bf vtm} = \tau(g^\omega(0))$ where
$g$ and $\tau$ are defined in Section~\ref{two}.   Let $h: \{0,1,2\}^* \rightarrow \Delta^*$ be a morphism.
If the three lengths $|h(0)|$,
$|h(1)|$, and $|h(2)|$ form an arithmetic
progression, then $h({\bf vtm})$ is $2$-automatic.
\label{thm1}
\end{theorem}

\begin{proof}
Suppose $a, b$ are integers, with
$a \geq 1$ and $a+2b \geq 1$, such
that $|h(i)| = a+ib$ for $i \in \{ 0,1,2 \}$.   Write ${\bf vtm} = c(0) c(1) c(2) \cdots$.
An easy induction now shows that
$$|h(c(0) c(1) \cdots c(n-1))| = (a+b)n + b t_n $$
for $n \geq 0$, where 
${\bf t} = t_0 t_1 \cdots$ is the Thue-Morse word.  To compute the $n$'th symbol of 
$h({\bf vtm})$, divide $n$ by $a+b$ to determine which block $h(c(i))$ it corresponds to; then adjust based on whether
$t_i = 0$ or not.   More precisely,
define $n' := \lfloor n/(a+b) \rfloor$ and
$m := n \bmod (a+b)$.  Then
\begin{displaymath}
(h({\bf vtm}))[n] := \begin{cases}
(h(c(n')))[m], &\text{if $t_{n'} = 0$;} \\
(h(c(n'-1)))[m+a+b], &\text{if $t_{n'} = 1$ and $t_{n'-1}= 0$ and $m < b$}; \\
(h(c(n')))[m-b], &\text{if $t_{n'} = 1$ and
$t_{n' - 1} = 0$ and $m \geq b$}; \\
(h(c(n'-1)))[m+a], &\text{if $t_{n'} = 1$ and $t_{n'-1}= 1$ and $m < b$}; \\
(h(c(n')))[m-b], &\text{if $t_{n'} = 1$ and $t_{n'- 1} = 1$ and $m \geq b$}.
\end{cases}
\end{displaymath}
For fixed $a$ and $b$,
an automaton on input $n$ in base $2$ can compute $n'$ and $m$ on the fly and do the required lookup.
\end{proof}

\begin{theorem}
The infinite word 
$\zeta({\bf vtm})$ contains only three distinct squares:
$0^2, 1^2,$ and $(01)^2$.  It is generated by an automaton with 88 states, and has weight is
$2 \cdot 88 = 176$.
\end{theorem}

\subsection{The Badkobeh-Crochemore construction}

Yet another construction was given by Badkobeh and Crochemore \cite{Badkobeh&Crochemore:2012,Badkobeh:2013}. 
They defined the morphism 
\begin{align*}
\xi(0) &= 000111 \\
\xi(1) &= 0011 \\
\xi(2) &= 01001110001101 \ .
\end{align*}
of weight $24$.
Although $\xi$ applied to a squarefree word can produce a word with more than three squares (consider $0102$), it turns out that
$\xi({\bf vtm})$ is squarefree.   Furthermore, although they do not mention it, $\xi$ is a morphism of lowest total weight with this property. 

Incidentally, we found another morphism with the same properties, of the same weight; it is
\begin{align*}
\kappa(0) &= 110100111000110100 \\
\kappa(1) &= 1100 \\
\kappa(2) &= 01 \ .
\end{align*}

However, the lengths of the images of both of these morphisms are not in arithmetic progression, and so Theorem~\ref{thm1} does not apply.   Indeed, we suspect (but did not prove) that neither $\xi({\bf vtm})$
nor $\kappa({\bf vtm})$ is a $2$-automatic sequence.  If they are $2$-automatic, then more than $200$ states are needed to generate them.

\subsection{Our first construction}

The previous section suggests looking for a morphism $\eta$ of lowest total weight, where the lengths of the images of $0,1,2$ are in arithmetic progression, such that $\eta({\bf vtm})$ has only $3$ distinct squares.
We found the following morphism, which is the smallest such, of weight $36$.
\begin{align*}
\eta(0) &= 00011101 \\
\eta(1) &= 001110001101 \\ 
\eta(2) &= 0011000111001101 \ .
\end{align*}

\begin{theorem}
The infinite word $\eta({\bf vtm})$ contains only
three distinct squares:  $0^2, 1^2$, and $(10)^2$.   It is $2$-automatic, and 
can be generated by an automaton of $27$ states, so its weight is $2 \cdot 27 = 54$.
\end{theorem}


\subsection{The last construction}

Finally, instead of using the strategy of applying a morphism to $\bf vtm$,  we can search directly for a $k$-automatic word of minimum total weight.   It turns out that this minimum weight is $44$, corresponding to a $2$-automaton with $22$ states:
\begin{figure}[h]
\includegraphics[width=7in]{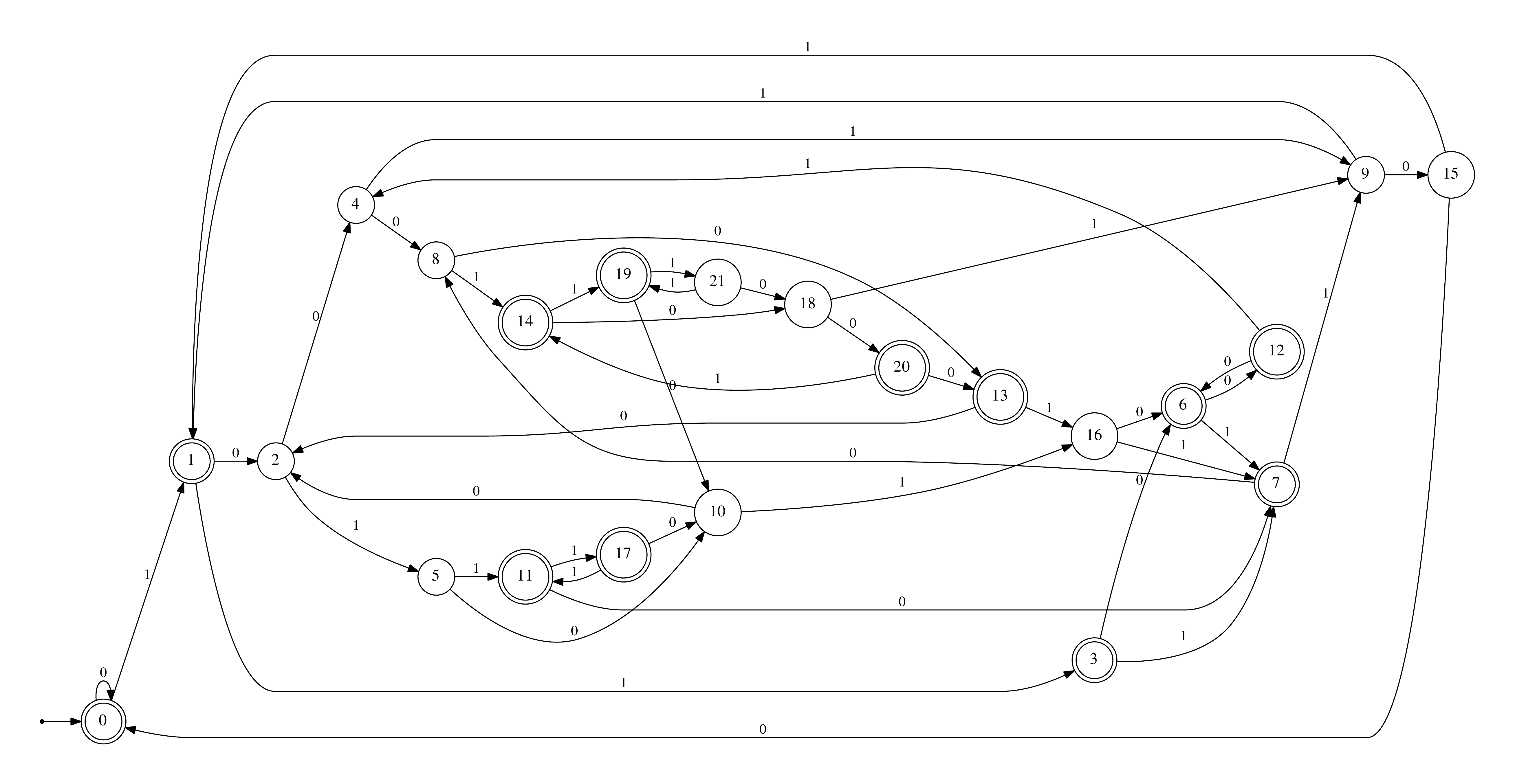}
\caption{DFAO where accepting states have output $1$ and all other states have output $0$.}
\label{figure:22}
\end{figure}

The corresponding representation is as the image, under the coding $\gamma$, of the fixed point of the morphism $q$ defined below over the alphabet $\{0,1,\ldots, 21\}$.   We use commas
to separate letters in the image of $q$, because of the large alphabet size.
\begin{align*}
q(0) &= 0,1 & \quad \gamma(0) &= 1 &
q(1) &= 2,3 & \quad \gamma(1) &= 1 \\
q(2) &= 4,5 & \quad \gamma(2) &= 0 &
q(3) &= 6,7 & \quad \gamma(3) &= 1 \\
q(4) &= 8,9 & \quad \gamma(4) &= 0 &
q(5) &= 10,11 & \quad \gamma(5) &= 0 \\
q(6) &= 12,7 & \quad  \gamma(6) &= 1 &
q(7) &= 8,9 & \quad \gamma(7) &= 1 \\
q(8) &= 13,14 & \quad \gamma(8) &= 0 &
q(9) &= 15,1  & \quad \gamma(9) &= 0 \\
q(10) &= 2,16 & \quad \gamma(10) &= 0 &
q(11) &= 7,17 & \quad \gamma(11) &= 1 \\
q(12) &= 6,4 & \quad \gamma(12) &= 1 &
q(13) &= 2,16 & \quad \gamma(13) &= 1 \\
q(14) &= 18,19 & \quad \gamma(14) &= 1 &
q(15) &= 0,1 & \quad \gamma(15) &= 0 \\
q(16) &= 6,7 & \quad \gamma(16) &= 0 &
q(17) &= 10,11 & \quad \gamma(17) &= 1 \\
q(18) &= 20,9 & \quad \gamma(18) &= 0 &
q(19) &= 10,21 & \quad \gamma(19) &= 1 \\
q(20) &= 13,14 & \quad \gamma(20) &= 1 &
q(21) &= 18,19 & \quad \gamma(21) &= 0
\end{align*}

\begin{theorem}
The infinite word
$$\gamma(q^\omega(0)) = 11010011000111001101001110001101000111010011000\cdots$$
contains only 
$3$ distinct squares:  $0^2$, $1^2$, and $(10)^2$. It has total weight $44$.
\end{theorem}

By exhaustive search we find that there are no $k$-automatic words containing only three distinct squares, with $s$ states, for $3 \leq k \leq 44$ and $ks \leq 44$.

We propose this word as the simplest of all binary words with three squares.

\section{Verifying the claims}

We used breadth-first search to find candidates for the minimal examples presented here.  The number of states in the minimal automaton were determined using the Myhill-Nerode theorem (see, e.g., \cite{Hopcroft&Ullman:1979}).    We used the theorem-proving software {\tt Walnut} \cite{Mousavi:2016} to verify assertions about the squares contained in each word.
For example, the claim about the $22$-state automaton in the previous section can be proved
as follows:   create the automaton, and call it {\tt Q} in {\tt Walnut}, and then evaluate the following three statements:
\begin{verbatim}
eval qtest1 "Ei,n (n>=3) & At (t<n) => Q[i+t]=Q[i+t+n]":
eval qtest2 "Ei (Q[i]=Q[i+1])&(Q[i]=Q[i+2])&(Q[i]=Q[i+3])":
eval qtest3 "Ei (Q[i]=@0)&(Q[i+1]=@1)&(Q[i+2]=@0)&(Q[i+3]=@1)":
\end{verbatim}
The first predicate  asserts that there is a square of order $\geq 3$ in the word.  The second asserts that there is a square of the form $(00)^2$ or $(11)^2$.  The third asserts that there is a square of the form $(01)^2$.   
Since all three queries return {\tt false}, the word has the desired properties.  The total computation time for this query is a few seconds on a laptop.

Each of Theorems 1,2,4,5,6 can be proved similarly, although some require significant memory resources and time.   The {\tt Walnut} code can be found on the website of the second author:\\
\centerline{\url{https://cs.uwaterloo.ca/~shallit/papers.html} \ .}


\begin{thebibliography}{10}

\bibitem{Badkobeh:2013}
G.~Badkobeh.
\newblock Infinite words containing the minimal number of repetitions.
\newblock {\em J. Discrete Algorithms} {\bf 20} (2013), 38--42.

\bibitem{Badkobeh&Crochemore:2012}
G.~Badkobeh and M.~Crochemore.
\newblock Fewest repetitions in infinite binary words.
\newblock {\em RAIRO Inform. Th\'eor. App.} {\bf 46} (2012), 17--31.

\bibitem{Berstel:1978}
J.~Berstel.
\newblock Sur la construction de mots sans {carr\'e}.
\newblock {\em {S\'eminaire} de {Th\'eorie} des {Nombres}}  (1978--1979),
  18.01--18.15.

\bibitem{Berstel:1995}
J.~Berstel.
\newblock {\em Axel {Thue's} Papers on Repetitions in Words: a Translation}.
\newblock Number~20 in Publications du Laboratoire de Combinatoire et
  d'Informatique {Math\'ematique}. Universit\'e du Qu\'ebec \`a Montr\'eal,
  February 1995.

\bibitem{Blanchet-Sadri&Currie&Rampersad&Fox:2013}
F.~Blanchet-Sadri, J.~Currie, N.~Rampersad, and N.~Fox.
\newblock Abelian complexity of fixed point of morphism $0 \mapsto 012$, $1
  \mapsto 02$, $2 \mapsto 1$.
\newblock {\em INTEGERS: Elect. J. of Combin. Number Theory} {\bf 14} (2014),
  \#A11 (electronic).

\bibitem{Charlier&Rampersad&Shallit:2012}
\'{E}milie Charlier, Narad Rampersad, and Jeffrey Shallit.
\newblock Enumeration and decidable properties of automatic sequences.
\newblock {\em Internat. J. Found. Comp. Sci.} {\bf 23} (2012), 1035--1066.

\bibitem{Cobham:1972}
A.~Cobham.
\newblock Uniform tag sequences.
\newblock {\em Math. Systems Theory} {\bf 6} (1972), 164--192.

\bibitem{Entringer&Jackson&Schatz:1974}
R.~C. Entringer, D.~E. Jackson, and J.~A. Schatz.
\newblock On nonrepetitive sequences.
\newblock {\em J. Combin. Theory Ser. A} {\bf 16} (1974), 159--164.

\bibitem{Fraenkel&Simpson:1994}
A.~S. Fraenkel and J.~Simpson.
\newblock How many squares must a binary sequence contain?
\newblock {\em Electronic J. Combinatorics} {\bf 2} (1994), {\#}R2.

\bibitem{Harju&Nowotka:2006}
T.~Harju and D.~Nowotka.
\newblock Binary words with few squares.
\newblock {\em Bull. European Assoc. Theor. Comput. Sci.} , No.\ 89, (2006),
  164--166.

\bibitem{Hopcroft&Ullman:1979}
J.~E. Hopcroft and J.~D. Ullman.
\newblock {\em Introduction to Automata Theory, Languages, and Computation}.
\newblock Addison-Wesley, 1979.

\bibitem{Mousavi:2016}
H.~Mousavi.
\newblock Automatic theorem proving in {{\tt Walnut}}.
\newblock Available at \url{http://arxiv.org/abs/1603.06017}, 2016.

\bibitem{Ochem:2006}
P.~Ochem.
\newblock A generator of morphisms for infinite words.
\newblock {\em RAIRO Inform. Th\'eor. App.} {\bf 40} (2006), 427--441.

\bibitem{Rampersad&Shallit&Wang:2005}
N.~Rampersad, J.~Shallit, and {M.-w.} Wang.
\newblock Avoiding large squares in infinite binary words.
\newblock {\em Theoret. Comput. Sci.} {\bf 339} (2005), 19--34.

\bibitem{Thue:1906}
A.~Thue.
\newblock {\"Uber} unendliche {Zeichenreihen}.
\newblock {\em Norske vid. Selsk. Skr. Mat. Nat. Kl.} {\bf 7} (1906), 1--22.
\newblock Reprinted in {\it Selected Mathematical Papers of Axel Thue}, T.
  Nagell, editor, Universitetsforlaget, Oslo, 1977, pp.~139--158.

\bibitem{Thue:1912}
A.~Thue.
\newblock {\"Uber} die gegenseitige {Lage} gleicher {Teile} gewisser
  {Zeichenreihen}.
\newblock {\em Norske vid. Selsk. Skr. Mat. Nat. Kl.} {\bf 1} (1912), 1--67.
\newblock Reprinted in {\it Selected Mathematical Papers of Axel Thue}, T.
  Nagell, editor, Universitetsforlaget, Oslo, 1977, pp.~413--478.

\end{thebibliography}
\end{document}